\documentclass[12pt]{amsart}

\usepackage{amsmath, amsthm, latexsym, amssymb, amsfonts, epsf, xcolor, hyperref}
\usepackage{caption, subcaption, enumerate, color}
\usepackage{multirow,multicol}
\usepackage{mathtools}
\usepackage{bm}
\usepackage{blkarray}
\usepackage{booktabs}
\usepackage{framed}
\usepackage{setspace}
\usepackage{tikz}
\usepackage{graphicx}
\usepackage{wrapfig}
\usepackage{float}
\usepackage[dvips]{epsfig}
\usepackage{bbm}
\usepackage{url}
\usepackage{booktabs}
\usepackage{bibentry}
\usepackage[normalem]{ulem}
\usepackage{comment}
\usepackage{lingmacros}
\usepackage{cleveref}
\usepackage{extarrows}
\usepackage{multirow}
\usepackage{mathtools}
\usepackage{enumerate}
\usepackage{array}
\usepackage{bigstrut}
\usepackage{exscale,relsize}
\usepackage{graphicx}
\usepackage[utf8]{inputenc}
\usepackage[linesnumbered, ruled]{algorithm2e}

\DeclarePairedDelimiter\abs{\lvert}{\rvert}%

\newcommand{\cal}[1]{\mathcal{#1}}

\newcommand{\cP}{\cal P}

\newcommand{\la}{\langle}
\newcommand{\ra}{\rangle}
\newcommand{\1}{\mathbbm{1}}
\newcommand{\II}{{\mathfrak{I}}}
\newcommand{\ev}{\operatorname{ev}}

\newcommand{\F}{{\mathbb F}}
\newcommand{\N}{{\mathbb N}}

\newcommand{\zn}{{\mathbb{Z}_{n}}}
\newcommand{\hzn}{{\hat{\mathbb{Z}}_{n}}}

\numberwithin{equation}{section}
\newtheorem{theorem}{Theorem}[section]
\newtheorem{lemma}[theorem]{Lemma}
\newtheorem{proposition}[theorem]{Proposition}
\newtheorem{corollary}[theorem]{Corollary}

\theoremstyle{definition}
\newtheorem{definition}[theorem]{Definition} 
\newtheorem{remark}[theorem]{Remark}
\newtheorem{example}[theorem]{Example}

\newcommand{\rmv}[1]{}

\DeclareMathOperator{\wt}{wt}

\DeclareMathOperator{\supp}{supp}
\DeclareMathOperator{\RM}{RM}
\DeclareMathOperator{\Amp}{Amp}

\newcommand{\rl}[1]{Lemma~\ref{L:#1}}
\newcommand{\rp}[1]{Proposition~\ref{P:#1}}

\newcommand{\rc}[1]{Corollary~\ref{C:#1}}
\newcommand{\rt}[1] {Theorem~\ref{T:#1}}

\theoremstyle{empty}

\begin{document}


\title[An algebraic characterization of binary CSS-T codes and cyclic CSS-T codes]{An algebraic characterization of binary CSS-T codes and cyclic CSS-T codes for quantum fault tolerance}

\author[E. Camps-Moreno]{Eduardo Camps-Moreno}
\address[Eduardo Camps-Moreno]{Department of Mathematics\\ Virginia Tech\\ Blacksburg, VA USA}
\email{eduardoc@vt.edu}

\author[H. L\'opez]{Hiram H. L\'opez}
\address[Hiram H. L\'opez]{Department of Mathematics\\ Virginia Tech\\ Blacksburg, VA USA}
\email{hhlopez@vt.edu}

\author[G. L. Matthews]{Gretchen L. Matthews}
\address[Gretchen L. Matthews]{Department of Mathematics\\ Virginia Tech\\ Blacksburg, VA USA}
\email{gmatthews@vt.edu}

\author[D. Ruano]{Diego Ruano}
\address[Diego Ruano]{IMUVA-Mathematics Research Institute\\ Universidad de Valladolid\\ Valladolid, Spain}
\email{diego.ruano@uva.es}

\author[R. San-Jos\'e]{Rodrigo San-Jos\'e}
\address[Rodrigo San-Jos\'e]{IMUVA-Mathematics Research Institute\\ Universidad de Valladolid\\ Valladolid, Spain}
\email{rodrigo.san-jose@uva.es}

\author[I. Soprunov]{Ivan Soprunov}
\address[Ivan Soprunov]{Department of Mathematics and Statistics\\ Cleveland State University\\ Cleveland, OH USA}
\email{i.soprunov@csuohio.edu}

\thanks{Hiram H. L\'opez was partially supported by the NSF grants DMS-2201094 and DMS-2401558.
Gretchen L. Matthews was partially supported by NSF DMS-2201075 and the Commonwealth Cyber Initiative.
Diego Ruano and Rodrigo San-Jos\'e were partially supported by Grant TED2021-130358B-I00 funded by MICIU/AEI/ 10.13039/501100011033 and by the ``European Union NextGenerationEU/PRTR'', by Grant PID2022-138906NB-C21 funded by MICIU/AEI/ 10.13039/501100011033 and by ERDF/EU, and by Grant QCAYLE supported by the European Union.-Next Generation UE/MICIU/PRTR/JCyL. Rodrigo San-Jos\'e was also partially supported by Grants FPU20/01311 and EST23/00777 funded by the Spanish Ministry of Universities.}

\begin{abstract}
CSS-T codes were recently introduced as quantum error-correcting codes that respect a transversal gate. A CSS-T code depends on a CSS-T pair, which is a pair of binary codes $(C_1, C_2)$ such that $C_1$ contains $C_2$, $C_2$ is even, and the shortening of the dual of $C_1$ with respect to the support of each codeword of $C_2$ is self-dual. In this paper, we give new conditions to guarantee that a pair of binary codes $(C_1, C_2)$ is a CSS-T pair. We define the poset of CSS-T pairs and determine the minimal and maximal elements of the poset. We provide a propagation rule for nondegenerate CSS-T codes. We apply some main results to Reed-Muller, cyclic, and extended cyclic codes. We characterize CSS-T pairs of cyclic codes in terms of the defining cyclotomic cosets. We find cyclic and extended cyclic codes to obtain quantum codes with better parameters than those in the literature.
\end{abstract}

\maketitle

\section{Introduction} \label{S:intro}
The development of large-scale, reliable quantum computing relies on quantum error correction to guard against the adverse impact of noise and decoherence. Quantum error-correcting codes were first discovered by Shor in 1995 \cite{Shor_95}. Soon after that, independent works by Calderbank and Shor \cite{CalderbankShor_96} and Steane \cite{Steane_96} outlined how classical linear codes could be used to construct quantum error-correcting codes, now referred to as CSS codes. The CSS construction uses a pair $(C_1, C_2)$ of classical linear codes, where the code $C_1$ contains the code $C_2$, to define a quantum stabilizer code. CSS codes are advantageous because they 
allow one to combine two appropriate classical codes into
a quantum stabilizer code. 
\rmv{
correct bit-flip and phase-shift errors independently, allowing them to protect against these errors separately, unlike more general stabilizer codes.} CSS codes have some nice properties, including propagation rules (see \cite{calderbankp,grasslp,rainsp} and the survey \cite{grasslsurvey}).

While generally not optimal, CSS codes are optimal among nondegenerate stabilizer codes that support the transversal $T$ gate; indeed it is demonstrated in \cite{rengaswamyOptimalityCSST} that for any non-degenerate stabilizer code that supports a physical transversal T gate, there is a CSS code with the same parameters that also does. CSS-T codes, introduced in \cite{calderbankclassicalcsst}, are motivated by the need for quantum codes which respect the transversal $T$ gate. Transversal gates are essential in fault-tolerant quantum computation as they mitigate the proliferation of errors. Transversals may be considered the most straightforward fault-tolerant realizations because they split into gates that act on individual qubits. 

A CSS-T code is formed using a pair $(C_1, C_2)$ of classical linear codes such that $C_1$ contains $C_2$, all codewords of $C_2$ are of even weight, and the shortening of the dual of $C_1$ with respect to the support of each codeword $c$ of $C_2$ is self-dual. In this case, we say that $(C_1, C_2)$ is a CSS-T pair. It is not surprising that it remains an open question to determine asymptotically good families of CSS-T codes \cite{albertocsst}. CSS-T codes from Reed-Muller codes have been explored in \cite{felicecsst}, and some general properties are laid out in \cite{albertocsst}.

In this paper, we study binary CSS-T pairs.
Section~\ref{S:equiv_defs} introduces the basic properties of CSS-T pairs. We give in Theorem~\ref{T:equiv-def} several conditions to determine if a pair of codes $(C_1, C_2)$ is a CSS-T pair. The equivalences of Theorem~\ref{T:equiv-def} allow us to see that the minimum distance of a CSS-T code associated with $(C_1, C_2)$ is lower bounded by the minimum distance of $C_2^\perp$. In Section~\ref{S:poset}, Corollary~\ref{C:monotone} allows us to define a poset $\cP$ of CSS-T pairs relative to the order $(C_1, C_2) \leq (C_1',C_2')$ if and only if $C_i\subset C_i'$ for $i=1,2$. We determine the minimal elements of $\cP$ in Corollary~\ref{minimal}. Using a sequence of results on properties of CSS-T pairs, we provide in Corollary~\ref{C:propagation} a propagation rule for nondegenerate CSS-T codes and characterize the maximal elements of $\cP$ in Theorem~\ref{T:maximal}. In Corollary~\ref{C:special}, we collect special cases when the conditions of \rt{maximal} can be relaxed. As an application, we apply some results of Section~\ref{S:poset} to Reed-Muller codes. In Section \ref{S:cyclic}, we restrict our attention to cyclic and extended cyclic codes. Theorem \ref{T:csstcyclic} provides a characterization of cyclic CSS-T pairs in terms of the defining cyclotomic cosets, and Corollary \ref{C:cycmaximal} characterizes those that are maximal. We find cyclic and extended cyclic codes that outperform binary Reed-Muller codes. In Section~\ref{S:tri} we compare our codes with triorthogonal codes \cite{bravyiTriorthogonalOriginal,haahClassificationTriorthogonal}. A summary and open problems are included in Section~\ref{S:conclusion}. Examples are provided throughout the paper. We conclude this section with a summary of results and a motivating example.

\subsection{Summary of major results}
In this subsection, we provide a guide to the major results of this paper. 

\begin{itemize}
    \item A primary contribution of this paper is the following more straightforward characterization of CSS-T pairs, found in Theorem~\ref{T:equiv-def}: 
Given binary linear codes $C_1$ and $C_2$ of length $n$, $$(C_1, C_2) \textnormal{ is a CSS-T pair if and only if } C_2\subset C_1\cap(C_1^{\star 2})^\perp.$$ 
Among the consequences are the fact that 
 $$C_2 \textnormal{ is self-orthogonal for all CSS-T pairs }(C_1, C_2).$$

\item Another key result is that CSS-T pairs form a poset $\cP$. According to Corollary \ref{C:monotone}, 
given a CSS-T pair $(C_1, C_2)$ $$(C_1',C_2) \textnormal{ is a CSS-T pair } \forall \ C_2\subset C_1'\subset C_1$$ and $$(C_1,C_2') \textnormal{ is a CSS-T pair } \forall \ C_2'\subset C_2.$$

\item  We demonstrate in Theorem \ref{T:maximal} that $$(C_1, C_2) \textnormal{ is a maximal CSS-T pair } \Leftrightarrow
C_1^\perp=C_1\star C_2 \textnormal{ and }
    C_2^\perp=C_1^{\star 2}.$$
Moreover, 
we determine minimal (Corollary \ref{minimal}) and maximal (Proposition \ref{P:max-in-C2} and Corollary \ref{C:max-in-C1}) elements of the poset $\cP$:  
$(C_1,C_2)$
is a maximal CSS-T pair 
\begin{itemize}
    \item 
with respect to  $C_2$ if and only if  $$C_2=C_1\cap(C_1^{\star 2})^\perp.$$ 
\item 
with respect to $C_1$ if and only if 
   $$   C_1=C_2^\perp\cap (C_1\star C_2)^\perp.$$ 
   \end{itemize}

\item Corollary~\ref{C:propagation} contains a propagation rule for nondegenerate CSS-T codes: 
Given a nondegenerate $[[n,k,d]]$ CSS-T code from a CSS-T pair $(C_1,C_2)$, for any $y\in C_2^\perp\cap (C_1\star C_2)^\perp \textnormal{ and } y\not \in C_1$, we have that $(C_1+\langle y\rangle,C_2)$ is a nondegenerate CSS-T pair with parameters
$
[[n,k+1,d]]$.

\item In Theorem \ref{T:csstcyclic}, we prove that for cyclotomic cosets
$I_1, I_2\subset \zn$, 
$$(C(I_1),C(I_2)) \textnormal{ is a CSS-T pair if and only if }
 I_2\subset I_1 \textnormal{ and }
  n\not \in (I_1+I_1+I_2).$$
The corresponding quantum code is a $[[n,\abs{I_1}-\abs{I_2},\geq n-\Amp(J_2)+1]]$ code.

\end{itemize}

\subsection{Motivating example}
We conclude this section with an example to demonstrate the utility of some of the results in the paper. In particular, we show how to apply them to the well known $[[15,1,3]]$ (punctured) quantum Reed-Muller code \cite{andersonquantumRM,quanQuantumRM}. Let $m\geq 1$ and $0\leq d \leq m-1$. Then the $d$-th order {\it binary Reed-Muller code} is defined as
$$
\RM_m(d):=\left\{ (f(v))_{v \in \F_2^m} : f \in \F_2[x_1, \dots, x_m], \deg f \leq d \right\}.
$$
Moreover, it is known that its dual code is $\RM_m(d)^\perp=\RM_m(m-1-d)$. Let $m=4$ and assume that we order the points in $\F_2^4$ so that $(0,0,0,0)$ corresponds to the first coordinate of the corresponding Reed-Muller codes. We consider $C_1=\RM_4(1)^{\{1\}}$, that is, the puncturing of the code $\RM_4(1)$ in the coordinate corresponding to $(0,0,0,0)$. For $C_2$, we consider the simplex code of length $15$. This corresponds to taking $C_2=\RM_4(1)_{\{ 1\}}$, the shortening of $\RM_4(1)$ in the first coordinate. The sets of monomials whose evaluation over $\F_2^4\setminus \{ (0,0,0,0)\}$ generates $C_1$ and $C_2$ are $\{ 1,x_1,x_2,x_3,x_4\}$ and $\{x_1,x_2,x_3,x_4\}$, respectively, and we have $C_2\subset C_1$. If we prove that $C_2\subset (C_1^{\star 2})^\perp$, then $C_2\subset C_1\cap (C_1^{\star 2})^\perp$, and, by Theorem \ref{T:equiv-def}, we would have that $(C_1,C_2)$ is a CSS-T pair. The Schur product $\RM_m(d_1)\star \RM_m(d_2)$, for some $0\leq d_1,d_2\leq m-1$, corresponds to taking the code generated by the evaluation of the products of the corresponding monomials. In this example, $C_1^{\star 2}=C_1\star C_1$ is the code generated by the evaluation over $\F_2^4\setminus \{ (0,0,0,0)\}$ of
$$
\{1, x_1,x_2,x_3,x_4,x_1x_2,x_1x_3,x_1x_4,x_2x_3,x_2x_4,x_3x_4\}.
$$
This actually corresponds to the puncturing in the first position of $\RM_4(2)$, that is, $C_1^{\star 2}=\RM_4(2)^{\{1\}}$. Since the dual of a puncturing is the corresponding shortening of the dual, we obtain $(C_1^{\star 2})^\perp =\RM_4(1)_{\{1\}}=C_2$. Thus, $(C_1,C_2)$ is a CSS-T pair. Analogously, one can prove that $C_1\star C_2$ is generated by 
$$
\{x_1,x_2,x_3,x_4,x_1x_2,x_1x_3,x_1x_4,x_2x_3,x_2x_4,x_3x_4\},
$$
that is, $C_1\star C_2=\RM_4(2)_{\{1\}}=C_1^\perp$. We proved before that $(C_1^{\star 2})^\perp =C_2$, which implies $C_1^{\star 2}=C_2^\perp$. By Theorem \ref{T:maximal}, we have that the $[[15,1,3]]$ (punctured) quantum Reed-Muller code is maximal with respect to the CSS-T poset $\mathcal{P}$.

\section{Equivalent Definitions}
\label{S:equiv_defs}
 
In this section, we give equivalent conditions for a pair of binary codes $(C_1, C_2)$ to be a CSS-T pair.

We start by fixing some notations for the rest of the paper. For a positive integer $n$, we write $[n]:=\{1, \dots, n \}$. We denote by $\1$ the element $(1,\ldots,1)$, where the number of entries depends on the context. We say a binary code $C$ of length $n$, dimension $k$, and minimum Hamming distance $d$ is an $[n,k,d]$ code. Let $C\subset \F_2^n$ be a code and $i\in [n]$. The \textit{dual} of $C$ with respect to the Euclidean inner product is denoted by $C^\perp$. The \textit{shortening} of $C$ in $\{i\}$, denoted by $C_{\{i\}}$, is the binary code
$$
C_{\{i\}} := \{ (c_1,\dots,c_{i-1},c_{i+1},\dots,c_n) : (c_1,\dots,c_{i-1},0,c_{i+1},\dots,n)\in C\}. 
$$
The \textit{puncturing} of $C$ in $\{i\}$, denoted by $C^{\{i\}}$, is the binary code
$$
C^{\{i\}} := \{ (c_1,\dots,c_{i-1},c_{i+1},\dots,c_n) : (c_1,\dots,c_{i-1},c_i,c_{i+1},\dots,c_n)\in C, \text{for some } c_i \in \F_2 \}. 
$$
For $S\subset[n]$, we write $C_S$ (resp. $C^S$) for the successive shortening (resp. puncturing) of $C$ in the coordinates indexed by the elements in $S$. 

The \textit{Schur product} of two vectors $x=(x_1,\ldots,x_n)$ and $y=(y_1,\ldots,y_n)$ in $\F_2^n$ is denoted and defined by
$$x \star y := (x_1y_1,\ldots, x_ny_n).$$
The \textit{Schur product} of two binary codes $C_1$ and $C_2$, denoted by $C_1\star C_2$, is defined as the binary code generated by the vectors
$$\left\{c_1 \star c_2 : c_i  \in C_i \right\}.$$
The $t$-\textit{fold Schur product} of $C$ with itself is $C^{\star t}:=\underbrace{C\star \cdots \star C}_{t}$, the $t$-th Schur power of $C$.
Note that for a binary code $C$, we always have $C\subset C^{\star 2}$ since $x\star x=x$ for any binary vector $x \in \F_2^n$.

Recall that a code is \textit{of even weight}, or \textit{even-weighted}, provided all of its codewords have even Hamming weight. For $x\in C$, we use $Z(x)$ to denote the set of positions of the zero coordinates of $x$, i.e., $Z(x)=[n]\setminus\supp(x)$, where $\supp(x)$ is the support of $x$ (set of nonzero entries of $x$).

We use $[[n,k,d]]$ to denote a quantum code that encodes $k$ logical qubits into $n$ physical qubits and can correct up to $d - 1$ erasures. We recall the CSS construction  \cite{CalderbankShor_96, Steane_96}.

\begin{theorem}[CSS Construction]\label{T:CSS}
Let $C_i\subset \F_{2}^n$ be linear codes of dimension $k_i$, for $i=1,2$, such that $C_2\subset C_1$. Then, there is an $[[n,k_1-k_2,d]]$ quantum code with 
$$
d=\min \left\{\wt\left(C_1\setminus C_2 \right), \wt\left(C_2^\perp\setminus C_1^\perp\right) \right\}.
$$
\end{theorem}

Let $d^*:=\min \{ \wt(C_1),\wt(C_2^\perp)\}$. If $d=d^*$, the corresponding quantum code is said to be \textit{nondegenerate}, and it is called \textit{degenerate} if $d>d^*$.

The following definition was given in \cite{calderbankclassicalcsst}.

\begin{definition}
    Let $C_2\subset C_1$ be binary codes. Then $(C_1, C_2)$ is a {\it CSS-T pair} if $C_2$ is even-weighted and for any $x\in C_2$, the shortening $(C_1^\perp)_{Z(x)}$ contains a self-dual code.
\end{definition}

\begin{theorem}\label{T:equiv-def} Let $C_1$ and $C_2$ be binary codes of length $n$.
    The following are equivalent.
    \begin{itemize}
        \item[\rm (1)] $(C_1, C_2)$ is a CSS-T pair.
        \item[\rm (2)] $C_2\subset C_1$, $C_2$ is even-weighted, and for any $x\in C_2$ the code $C_1^{Z(x)}$ is self-orthogonal.
        \item[\rm (3)] $C_2\subset C_1\cap(C_1^{\star 2})^\perp$.
        \item[\rm (4)] $C_1^\perp+C_1^{\star 2}\subset C_2^\perp$.
    \end{itemize}
Moreover, if $(C_1, C_2)$ is a CSS-T pair then $C_2$ is self-orthogonal.
\end{theorem}

\begin{proof}
The equivalence of (1) and (2) was proved in \cite{felicecsst}. (See also \cite{albertocsst} for the case of arbitrary fields of characteristic 2.) Also, (3) and (4) are equivalent by taking the duals. 
    
To show the equivalence of (2) and (3), note that for any $x\in C_2$, the code $C_1^{Z(x)}$ is self-orthogonal if and only if $x\in (C_1^{\star 2})^\perp$. Indeed, $x\in (C_1^{\star 2})^\perp$ if and only if $\sum_{i=1}^n x_iu_iv_i=0$ for any $u,v\in C_1$. As $x$ is a binary vector, we can write this as $\displaystyle\sum_{i\in \supp(x)} u_iv_i=0$, i.e., $u'\cdot v'=0$ for any $u',v'\in C_1^{Z(x)}$, that is $C_1^{Z(x)}$ is self-orthogonal. On the other hand, if $C_2\subset C_1\cap (C_1^{\star 2})^\perp$, then we have
$$C_2\subset C_1\subset C_1^{\star 2}\subset C_2^\perp.$$
Thus, $C_2$ is even-weighted because it is self-orthogonal.
\end{proof}

\begin{remark}
Note that if $(C_1, C_2)$ is a CSS-T pair then, by part (4) of Theorem~\ref{T:equiv-def}, $C_1^{\star 2}\subset C_2^\perp$, which is equivalent to $C_1\star C_2\subset C_1^\perp$. This observation previously appeared in \cite[Remark 3]{calderbankclassicalcsst}.
\end{remark}

A {\it CSS-T code} is a code obtained via a CSS-T pair and Theorem~\ref{T:CSS}. The equivalences of Theorem~\ref{T:equiv-def} allow us to see some structural properties of CSS-T codes. In particular, the minimum distance of a CSS-T code associated with $(C_1, C_2)$ is lower bounded by the minimum distance of $C_2^\perp$.

\begin{corollary}\label{C:dmin2perp}
Let $(C_1, C_2)$ be a CSS-T pair. Then 
$$
\min \{\wt(C_1),\wt(C_2^\perp)\}=\wt(C_2^\perp),
$$
and the parameters of the corresponding CSS-T code are $$[[n,k_1-k_2,\ge \wt(C_2^\perp)]].$$
Moreover, if the code is nondegenerate, we have equality in the minimum distance.
\end{corollary}
\begin{proof}
From Theorem \ref{T:equiv-def}~(4), we see that
$$
\wt(C_2^\perp)\leq \wt(C_1^\perp+C_1^{\star 2})\leq \wt(C_1^{\star 2})\leq \wt(C_1).
$$ 
\end{proof}

\section{The poset of CSS-T pairs}
\label{S:poset}

Let $(C_1, C_2)$ be a CSS-T pair. By Corollary \ref{C:dmin2perp}, the CSS-T code associated with the pair $(C_1, C_2)$ has parameters $[[n,k_1-k_2,\ge \wt(C_2^\perp)]]$. Thus, increasing the dimension of $C_1$ will increase the dimension of the associated CSS-T code, and the minimum distance is still bounded by $\wt(C_2^\perp)$. In particular, if the associated CSS-T code is nondegenerate, then increasing the dimension of $C_1$ does not change the minimum distance (see Corollary \ref{C:dmin2perp}). On the other hand, increasing the dimension of $C_2$ could improve the minimum distance but decrease the dimension of the resulting CSS-T code.

The following Corollary allows us to define a partial order on the set of CSS-T pairs. The result shows that all the CSS-T pairs are determined by those CSS-T pairs $(C_1, C_2)$ that cannot be extended to another CSS-T pair $(C_1^\prime, C_2^\prime)$, where $C_1=C_1^\prime$ or $C_2=C_2^\prime$.

\begin{corollary}\label{C:monotone}
Let $(C_1, C_2)$ be a CSS-T pair. Then, the following hold.
\begin{itemize}
    \item[\rm (1)] $(C_1',C_2)$ is a CSS-T pair for any $C_2\subset C_1'\subset C_1$.
    \item[\rm (2)] $(C_1,C_2')$ is a CSS-T pair for any $C_2'\subset C_2$.
\end{itemize}
\end{corollary}
\begin{proof}
(1) As $C_1'\subset C_1$, then $(C_1'^\perp)_{Z(x)}\supset (C_1^\perp)_{Z(x)}$ for any $x\in C_2$. Hence, if $(C_1^\perp)_{Z(x)}$ contains a self-dual code, then  $(C_1'^\perp)_{Z(x)}$ also contains a self-dual code.

(2) It is a direct consequence of \rt{equiv-def}~(2).
\end{proof}
We are ready to define a partial order in the set of CSS-T pairs.
\begin{definition}
We denote by $\cP$ the \textit{poset} of CSS-T pairs relative to the order $(C_1, C_2) \leq (C_1',C_2')$ if and only if $C_i\subset C_i'$ for $i=1,2$.
\end{definition}
From now on, we discard the trivial pairs $(C_1,\{0\})$ from $\cP$. Denote by $\la x\ra$ the code generated by an element $x \in \F_2^n$.
\begin{corollary}\label{minimal}
The set of minimal elements of $\cP$ is
$$\left\{ (\la u\ra, \la u\ra) : u \text{ even }, u \in \F_2^n \right\}.$$
\end{corollary}
\begin{proof}
This is a consequence of Corollary~\ref{C:monotone}.
\end{proof}
We are interested in the set of maximal elements of $\cP$.
\begin{definition}
We say that $(C_1, C_2)\in\cP$ is {\it maximal in $C_1$} if $(C_1, C_2)\leq (C_1',C_2)$ implies $C_1=C_1'$. Similarly, $(C_1, C_2)$ is {\it maximal in $C_2$} if $(C_1, C_2)\leq (C_1,C_2')$ implies $C_2=C_2'$.
\end{definition}
Note that a pair $(C_1, C_2)$ is a maximal element of $\cP$ if and only if $(C_1, C_2)$ is maximal in both $C_1$ and $C_2$. Some maximal elements in $\cP$ are given by the pairs $(C_1, C_2)$ where $C_1$ has codimension one. Indeed, by \rt{equiv-def}~(4), $C_1^{\star 2}\subset C_2^\perp$. Since we assume that $C_2$ is nontrivial, we see that  $C_1^{\star 2}$ is a proper subspace of $\F_2^n$, obtaining thus that $C_1=C_1^{\star 2}=C_2^\perp$. Hence, $C_2$ is a one-dimensional subspace of $C_1$ generated by an even-weight vector. In fact, we show in \rt{maximal} that the property $C_1^{\star 2}=C_2^\perp$ holds for any maximal pair $(C_1, C_2)$.
 
We start by describing pairs that are maximal in $C_2$.

\begin{proposition}\label{P:max-in-C2}
A pair $(C_1, C_2)\in\cP$ is maximal in $C_2$ if and only if $C_2=C_1\cap(C_1^{\star 2})^\perp$.
\end{proposition}
\begin{proof}
This is provided by \rt{equiv-def}~(3).
\end{proof}
The following proposition gives a criterion for extending a CSS-T pair $(C_1, C_2)$ to a pair $(C_1', C_2)$ with $\dim C_1'=\dim C_1+1$. 

\begin{proposition}\label{P:step}
    Let $(C_1, C_2)$ be a CSS-T pair and $y \in \F_2^n$. Then 
    $(C_1+\la y\ra,C_2)$ is a CSS-T pair if and only if $C_1\star y+\la y\ra\subset C_2^\perp$, or equivalently, $y\in C_2^\perp\cap (C_1\star C_2)^\perp$.
\end{proposition}
\begin{proof}
Define $C_1' := C_1+\la y\ra$. Note that $C_1'^\perp\subset C_1^\perp$. Since $(C_1, C_2)$ is a CSS-T pair, we have $C_1^\perp+C_1^{\star 2}\subset C_2^\perp$ by \rt{equiv-def}~(4). Thus,
$$C_1'^\perp\subset C_1^\perp \subset C_1^\perp+C_1^{\star 2}\subset C_2^\perp.$$

By \rt{equiv-def}~(4), $(C_1',C_2)$ is a CSS-T pair if and only if $C_1'^\perp+C_1'^{\star 2}\subset C_2^\perp$. So, it is enough to verify $C_1'^{\star 2}\subset C_2^\perp$ if and only if $C_1\star y+\la y\ra\subset C_2^\perp$.
    It remains to notice that $C_1'^{\star 2}=C_1^{\star 2}+C_1\star y+\la y\ra$, as $y\star y=y$.
\end{proof}

Unlike Proposition~\ref{P:max-in-C2}, Proposition \ref{P:step} does not allow us to find the maximal $C_1$ for a given $C_2$ to get a CSS-T pair as the next example shows.
        
\begin{example}\rm
    Let $C=\langle (1,1,1,1,1,1)\rangle$. By Proposition \ref{minimal}, $(C,C)\in\mathcal{P}$ and it is a minimal element. We have $C^\perp\cap(C^{\star 2})^\perp=C^\perp$. Let $v=(1,1,1,1,0,0), w=(1,0,0,0,0,1)\in C^\perp$. Thus $(C+\langle v\rangle, C)\in \mathcal{P}$, but $(C+\langle v,w\rangle,C)\notin\mathcal{P}$, despite $v,w\in C^\perp$.

    We have: $$C^\perp\cap((C+\langle v\rangle)\star C)^\perp=\langle (1,1,0,0,0,0), (1,0,1,0,0,0),(1,0,0,1,0,0),(0,0,0,0,1,1)\rangle.$$ 
    
    We can take any non-zero element $v'$ different from $(1,1,1,1,1,1,1)$ in this intersection and we get that $(C+\langle v,v'\rangle, C)$ is a CSS-T pair. Note that for $v'$ equal to $(1,1,0,0,0,0)$, $(1,0,1,0,0,0)$, or $(1,0,0,1,0,0)$, we get a new CSS-T pair. However, we do not obtain a new CSS-T for $v' = (0,0,0,0,1,1)$ since $v' \in C + \langle v\rangle$.
\end{example}

\begin{remark}\label{R:addone}
Note that, if $(C_1, C_2)$ is a CSS-T pair, then so is $(C_1+\la\1\ra, C_2)$. This follows from \rt{equiv-def}~(3), the previous result, and the observation that $C_2\subset\la\1\ra^\perp$, as $C_2$ is even-weighted. 
\end{remark}

Proposition \ref{P:step} also provides the following propagation rule for nondegenerate CSS-T codes. 

\begin{corollary}\label{C:propagation}
Let $(C_1, C_2)$ be a CSS-T pair such that the associated $[[n,k,d]]$ CSS-T code is nondegenerate. For any $y\in C_2^\perp\cap (C_1\star C_2)^\perp$ and $y\not \in C_1$, the pair $(C_1+\langle y\rangle,C_2)$ is a nondegenerate CSS-T pair with parameters
$$
[[n,k+1,d]].
$$
\end{corollary}
\begin{proof}
By Proposition \ref{P:step}, $(C_1+\langle y\rangle, C_2)$ is a CSS-T pair, and the parameters follow from Corollary \ref{C:dmin2perp}.
\end{proof}

\begin{corollary}\label{C:max-in-C1}
    A pair $(C_1, C_2)\in\cP$ is maximal in $C_1$ if and only if
    $C_1=C_2^\perp\cap (C_1\star C_2)^\perp$. 
\end{corollary}

\begin{proof}
By \rp{step}, $(C_1, C_2)\in\cP$ is maximal in $C_1$ if and only $C_2^\perp\cap (C_1\star C_2)^\perp\subset C_1$. On the other hand, the pair $(C_1+\la y\ra,C_2)$ is CSS-T for each $y\in C_1$, so by \rp{step}, $C_1\subset C_2^\perp\cap (C_1\star C_2)^\perp$ as well.
\end{proof}

We obtain the following theorem by combining the previous results on maximality in $C_1$ and $C_2$.

\begin{theorem}\label{T:maximal}
Let $C_2\subset C_1\subset\mathbb{F}_2^n$ be linear codes. The pair $(C_1, C_2)$ is maximal in $\cP$ if and only if
\begin{enumerate}
    \item[\rm (1)] $C_1^\perp=C_1\star C_2$ and
    \item[\rm (2)] $C_2^\perp=C_1^{\star 2}$.
\end{enumerate}
\end{theorem}

\begin{proof}
Assume $(C_1, C_2)$ is a maximal CSS-T pair. Note that we can assume $\1 \in C_1$ by Remark \ref{R:addone}, and we have $C_2=\la\1\ra\star C_2\subset C_1\star C_2$. Now, by \rc{max-in-C1}, we have 
    $$C_1^\perp=C_2+C_1\star C_2=C_1\star C_2,$$
which shows (1).   

As $C_2=C_1\cap (C_1^{\star 2})^\perp$ by \rp{max-in-C2}, we only need to show that $(C_1^{\star 2})^\perp \subset C_1$ in order to prove (2). Since $C_2\subset C_1$, we have $C_1\star C_2\subset C_1^{\star 2}$ and $(C_1^{\star 2})^\perp\subset (C_1\star C_2)^\perp$. Also, $C_2\subset C_1\subset C_1^{\star 2}$ implies that $(C_1^{\star 2})^\perp\subset C_2^\perp$. Therefore, by \rc{max-in-C1}, we get
$$(C_1^{\star 2})^\perp\subset C_2^\perp\cap (C_1\star C_2)^\perp= C_1.$$

Theorem \ref{T:equiv-def} (2) implies that $(C_1,C_2)$ is a CSS-T pair. The maximality follows directly from \rp{max-in-C2} and \rc{max-in-C1}, using both (1) and (2).
\end{proof}

The following example illustrates that the necessary condition (2) of \rt{maximal} for $(C_1, C_2)$ to be maximal is not sufficient.

\begin{example}\label{ex.hoy}
Define $C_2 := \la (1, 1, 0,0,0,0)\ra$ and
$C_1$ as the code whose generator matrix is given by
$$\left(
\begin{matrix}
1 & 1 & 0 & 0 & 0 & 0\\
0 & 0 & 1 & 1 & 1 & 0\\
0 & 0 & 0 & 1 & 0 & 1\\
0 & 0 & 0 & 0 & 1 & 1
\end{matrix}\right).$$
It is not difficult to see using \cite{cod_package, Mac2} that a generator matrix for $C_1^{\star 2}$ is given by
$$\left(
\begin{matrix}
1 & 1 & 0 & 0 & 0 & 0\\
0 & 0 & 1 & 0 & 0 & 0\\
0 & 0 & 0 & 1 & 0 & 0\\
0 & 0 & 0 & 0 & 1 & 0\\
0 & 0 & 0 & 0 & 0 & 1
\end{matrix}\right).$$
Hence, $(C_1^{\star 2})^\perp=\la (1, 1, 0,0,0,0)\ra=C_2$, meaning that the pair
$(C_1, C_2)$ satisfies condition (2) of \rt{maximal}. But the pair $(C_1, C_2)$ is not maximal in $C_1$ because the extension $(C_1+\la\1\ra,C_2)$ satisfies (1)--(2) of \rt{maximal}, meaning that it is maximal.
\end{example}

In the following Corollary, we collect special cases when the conditions of \rt{maximal} can be relaxed.

\begin{corollary}\label{C:special}
Let $C$ be a binary code.
\begin{enumerate}
\item[\rm (1)] The pair $(C,C)$ is maximal in $\cP$ if and only if $C^{\star 2}=C^\perp$.
\item[\rm (2)] If $C^\perp\subset C$, the pair $(C,C^\perp)$ is maximal in $\cP$ if and only if $C^{\star 2}=C$. Equivalently, $C$ is generated by vectors with pair-wise disjoint support.
\end{enumerate}
\end{corollary}

\begin{proof}
(1) If the pair $(C,C)$ is maximal in $\cP$, then $C^{\star 2}=C^\perp$ by Theorem~\ref{T:maximal}~(2).
If $C^{\star 2}=C^\perp$, then $(C, C)$ is a CSS-T pair by \rt{equiv-def}~(3). Also, the pair $(C, C)$ is maximal in $\cP$ by Theorem~\ref{T:maximal}.

(2) 
If $(C,C^\perp)$ is a maximal CSS-T pair, then $C=C^{\star 2}$ by \rt{maximal}~(2). Conversely, assume that $C=C^{\star 2}$. \rt{equiv-def}~(3) verifies that $(C, C^\perp)$ is a CSS-T pair. \rp{max-in-C2} verifies that $(C, C^\perp)$ is maximal in $C^\perp$. If $(C +\la y\ra, C^\perp)$ is a CSS-T pair for some $y\in\F_2^n$, then $y\in C$ by \rp{step}, meaning that $(C, C^\perp)$ is maximal in $C$.
\end{proof}

\begin{example}
Assume $3d=m-1$ for some $d,m\in\N$. For the binary Reed-Muller code $C := \RM_m(d)$, we have 
    $$C^\perp=\RM_m(d)^\perp=\RM_m(m-d-1)=\RM_m(2d)=C^{\star 2}.$$
Thus, $(C, C)$ is a maximal pair by \rc{special}~(1).
\end{example}

Observe that even if $(C_1,C_2)$ is maximal in $\mathcal{P}$, in principle, there can be a pair $(D_1,D_2)\in\mathcal{P}$ such that $C_2\subset D_2$ or $C_1\subset D_1$. We can give a complete characterization of such spaces. First we need a lemma.

\begin{lemma}\label{L:l}
    Let $C\subsetneq\mathbb{F}_2^n$ such that for any $x\in C\cap(C^{\star 2})^\perp$ we have $C\star x=C^\perp$. Then $(C^{\star 2})^\perp=\langle y\rangle$, for some $y\in C$, or $C=C^\perp$ and $C^{\star 2}=\langle\1\rangle^\perp$.
\end{lemma}

\begin{proof}
    First observe that $C\star x = C^\perp\subset C^{\star 2}$ implies $(C^{\star 2})^\perp\subset C$ and thus $C\cap (C^{\star 2})^\perp=(C^{\star 2})^\perp$. Let $y\in(C^{\star 2})^\perp$ be a minimal support codeword. If $y=\1$, then $C\star y=C=C^\perp$ and $C^{\star 2}=\langle \1\rangle^\perp$.

    Assume now that $\wt(y)<n$. Since $C\star y=C^\perp$ then $\langle e_i\ :\ i\notin\mathrm{supp}(y)\rangle\subseteq C$. If there is another minimal codeword $y\neq x\in(C^{\star 2})^\perp$, the same arguments lead to the existence of $i\in \mathrm{supp}(y)\setminus\mathrm{supp}(x)$ such that $e_i\in C^{\star 2}$ and thus $z_i=0$ for any $z\in(C^{\star 2})^\perp$, which contradicts that $y_i\neq 0$. Thus, there are no more minimal codewords in $(C^{\star 2})^\perp$ and we have the conclusion.
\end{proof}

The next example shows that the converse of the last lemma is not true.

\begin{example}\rm
    Let $C=\langle (1,1,0,0,0),(0,1,1,0,0),(0,0,0,1,1)\rangle$. We have $$C^{\star 2}=\langle (1,0,0,0,0),(0,1,0,0,0),(0,0,1,0,0),(0,0,0,1,1)\rangle,$$ and $(C^{\star 2})^\perp=\langle (0,0,0,1,1)\rangle$. However,

    $$C\star (0,0,0,1,1)=(0,0,0,1,1)\subsetneq C^\perp=\langle (1,1,1,0,0),(0,0,0,1,1)\rangle.$$
\end{example}

\begin{proposition}
    Let $(C_1,C_2)\in\mathcal{P}$. Then
    \begin{enumerate}
        \item There is no $(D_1,D_2)\in\mathcal{P}$ with $C_1\subsetneq D_1$ if and only if $C_1^\perp=C_1\star y$ for any $y\in C_1\cap (C_1^{\star 2})^\perp$.
        \item There is no $(D_1,D_2)\in\mathcal{P}$ with $C_2\subsetneq D_2$ if and only if $(C_2,C_2)$ is maximal.
    \end{enumerate}
\end{proposition}

\begin{proof}
    If there is no such $D_1$, since for any $y\in C_1\cap (C_1^{\star 2})^\perp$, $(C_1,\langle y\rangle)\in\mathcal{P}$ but $C_1$ cannot be extended, then $C_1=\langle y\rangle^\perp\cap (C_1\star y)^\perp=(C_1\star y)^\perp$ by Corollary \ref{C:max-in-C1} (note that $y \in C_1$ implies $y\in C_1\star y$). On the other hand, assume $C_1^\perp=C_1\star y$ for any $y\in C_1\cap (C_1^{\star 2})^\perp$, and let $C_1\subset D_1$ such that $D_1$ is the largest code containing $C_1$ with $(D_1,D)\in\mathcal{P}$ for some $D$. By the first part of this proof, the hypothesis and Lemma \ref{L:l} we have $(D_1^{\star 2})^\perp\subseteq(C_1^{\star 2})^\perp=\langle y\rangle$ for some $y\in C_1$. This implies $D_1\cap (D_1^{\star 2})^\perp =\langle y\rangle$ because $(D_1,D)\in \mathcal{P}$. By the choice of $D_1$ and the first part of the proof, $D_1\star y =D_1^\perp$, and we also have $C_1\star y=C_1^\perp$. Thus,
    
    $$C_1\star y\subset D_1\star y = D_1^\perp \subset C_1^\perp\Rightarrow D_1^\perp=C_1^\perp,$$

    \noindent and we get $D_1=C_1$. 

    To prove (2), observe that $(D_1,D_2)\in\mathcal{P}$ is such that $C_2\subset D_2$ if and only if there is $y\notin C_2$ such that $(C_2+\langle y\rangle,C_2)\in\mathcal{P}$ by Corollary \ref{C:monotone}. This happens if and only if $y\in (C_2^\perp \cap (C_2^{\star 2})^\perp)\setminus C_2$ by Proposition \ref{P:step}. However, $(C_2^{\star 2})^\perp\subset C_2^\perp$ and thus, $y\in (C_2^{\star 2})^\perp \setminus C_2$. If there is not such $y$, it means that $(C_2^{\star 2})^\perp = C_2$ and by Corollary \ref{C:special} we have the conclusion.
\end{proof}

\begin{example}\rm
    Let $$G=\begin{pmatrix} 1&1&1&1&1&1&1&1\\
    1&1&1&1&0&0&0&0\\
    1&1&0&0&1&1&0&0\\
    1&0&1&0&1&0&1&0\end{pmatrix}$$

    and $C$ be the code generated by $G$. We can check that $C^{\star 2}=\langle\1\rangle^\perp$, $C=C^\perp$ and thus, $(C,\langle \1\rangle)\in\mathcal{P}$ and there is no other CSS-T pair $(D_1,D_2)$ with $C_1\subsetneq D_1$.
\end{example}

\begin{corollary}
    If $(C_1,C_2)\in\mathcal{P}$ and there is no $D_1\supsetneq C_1$ and $D_2$ such that $(D_1,D_2)\in\mathcal{P}$, then for some $y\in C_1$, $C_2=\langle y\rangle$ and $(C_1,C_2)$ is maximal.
\end{corollary}

\section{Cyclic codes}
\label{S:cyclic}

We now illustrate the results from the previous sections using cyclic codes (and extended cyclic codes). We will review cyclic codes over $\F_q$, but note that we restrict to the case $q=2$ whenever we refer to CSS-T codes.

Take an integer $s>1$ and consider the field extension $\mathbb{F}_{q^s} / \mathbb{F}_q$. We set $n$ with $n\mid q^s-1$ and $g\in \mathbb{F}_q[x]$ such that $g$ divides $x^n-1$. We denote by $C_g$ the cyclic code with $g$ as its generator polynomial. Let $\beta\in\F_{q^s}$ be a primitive $n$-th root of unity. For the set $\zn:=\mathbb{Z}/n\mathbb{Z}$, we will consider the representatives between $1$ and $n$, i.e., $\zn=\{1,2,\dots,n\}$. 

\begin{definition}
The \textit{defining set} is given by $J := \{j\in \zn : g(\beta^j)=0\}$ and the \textit{generating set} by $I:=\{i\in \zn : g(\beta^i)\neq 0\}$.
\end{definition}

Note that $J=[n]\setminus I$, and 
$$
g=\prod_{j\in J}(x-\beta^j)=\frac{x^n-1}{\prod_{i\in I}(x-\beta^i)}.
$$

Define $-I:=\{n-i : i \in I\}\subset \zn.$ Let $\mathcal{M}\subset \mathbb{Z}_{\geq 0}$ be a finite set. We consider the $\F_{q^s}$-linear subspace
$$
\mathcal{L}(\mathcal{M}):=\langle x^i : i\in \mathcal{M} \rangle \subset  \mathbb{F}_{q^s}[x].
$$
Take a set of points $X=\{P_1,\dots,P_{\abs{X}}\}\subset \F_{q^s}$. We can define the following evaluation map associated to $X$:
$$
\begin{array}{cccl}
\ev_X \colon &\F_{q^s}[x]&\to& \F_{q^s}^{\abs{X}} \\
&f &\mapsto& \left(f(P_1),\dots,f(P_{\abs{X}})\right).
\end{array}
$$
Let $X_n:=\{1,\beta,\dots,\beta^{n-1}\}$, i.e., $X_n$ is the zero locus of $x^{n}-1$ in $\F_{q^s}$. We now consider the associated evaluation code

$$
B(\mathcal{M}):=\ev_{X_n}(\mathcal{L}(\mathcal{M}))=\{(f(1),f(\beta),\dots,f(\beta^{n-1})) : f\in  \mathcal{L}(\mathcal{M})\}\subset \mathbb{F}_{q^s}^n,
$$
and we define
$$
C(I):=B(-I)\cap \mathbb{F}_{q}^n.
$$
From \cite{bierbrauercyclic}, we obtain that $C_g=C(I)$, i.e., we have a description of cyclic codes in terms of subfield subcodes of evaluation codes. 

The definitions clearly show that $J$ and $I$ are closed under multiplication by $q$, which leads to the following definition.

\begin{definition}
Given a subset $I\subset \zn$, denote $q\cdot I : =\{q\cdot i : i \in I\}$. We say that $I$ is a \textit{cyclotomic coset} if $I=q\cdot I$. Let $a\in \zn$, the set $\II_a:=\{q^j \cdot a : j\geq 0 \}\subset \zn$ is the \textit{minimal cyclotomic coset} associated to $a$. 
\end{definition}

\begin{example}\label{EX:cyc1}
Let $q=2$, $s=4$, and $n=15$. Then, the minimal cyclotomic cosets are
$$
\II_1=\{1,2,4,8\},\; \II_3=\{3,6,12,9\},\; \II_5=\{5,10\}, \; \II_7=\{7,14,13,11\},\; \II_{15}=\{15\}.
$$
\end{example}

From \cite{bierbrauercyclic}, we have the following result about the dual of a cyclic code. 

\begin{theorem}\label{T:dualcyclic}
Let $I\subset \zn$ be a cyclotomic coset. We have that
$$
C(I)^\perp=C(-J).
$$
\end{theorem}

This last result can be seen as a consequence of the following fact from \cite{bierbrauercyclic}: If $I$ is a cyclotomic coset, then
\begin{equation}\label{commutativity}
(B(-I)\cap \F_{q}^n )^\perp= (B(-I)^\perp)\cap \F_{q}^n.
\end{equation}

The length of $C(I)$ is $n$, and its dimension is $\abs{I}$. For the minimum distance, we need the following definition.

\begin{definition}
The {\it amplitude} of a nonempty subset $I\subset \zn$ is
$$
\Amp(I):=\min \{i\in \N : \exists c\in  \zn \text{ such that } I\subset\{c,c+1,\dots,c+i-1 \} \}.
$$
\end{definition}

Then, the minimum distance of $C(I)$ is greater than or equal to $n-\Amp(I)+1$; for example, see \cite{cascudosquarescyclic}. Summarizing, $C(I)$ has parameters 
$$
[n,\,\abs{I},\,\geq n-\Amp(I)+1].
$$
Since $\Amp(-J)=\Amp(J)$, we see that $C(I)^\perp$ has parameters $[n,\,\abs{J},\,\geq n-\Amp(J)+1]$. Note that $n-\Amp(J)+1$ is equal to the usual BCH bound, i.e., it is equal to $\delta(I)+1$, where $\delta(I)$ is the maximum number of consecutive elements in $I$.

Given $I_1,I_2\subset \zn$, we consider their Minkowski sum
\begin{equation}\label{sumcyclo}
I_1+I_2:=\{ i_1+i_2 : i_1\in I_1,\; i_2\in I_2 \}\subset \zn.
\end{equation}

It is easy to check that if $I_1, I_2\subset \zn$ are cyclotomic cosets, then $I_1+I_2$ is also a cyclotomic coset. Following the previous notation, we will denote $J_i=[n]\setminus I_i$, for $i=1,2$.

\begin{example}\label{EX:cyc2}
Continuing with Example \ref{EX:cyc1}, we consider 
$$
I_1=\{1,2,4,8,15\},\; I_2=\{1,2,4,8\}.
$$
We compute the following Minkowski sums, which we will use in the following examples:
$$
I_1+I_2=\{1,2,3,4,5,6,8,9,10,12\}, \; I_1+I_1=(I_1+I_2)\cup \{15\}.
$$
Note that $I_1+I_2=\II_1\cup \II_3\cup\II_5$, i.e., $I_1+I_2$ is also a cyclotomic coset.
\end{example}

The following result from \cite{cascudosquaresmpc} shows that the sum and the Schur product of cyclic codes is also a cyclic code. 

\begin{lemma}\label{L:sumschur}
Let $I_1$ and $I_2$ be cyclotomic cosets. Then
$$
\begin{aligned}
&C(I_1)+C(I_2)=C(I_1 \cup I_2),\\
&C(I_1)\star C(I_2)= C(I_1+I_2).
\end{aligned}
$$
\end{lemma}

As an application of \rt{equiv-def}, we obtain the following criterion for a pair of cyclic codes to be a CSS-T pair.

\begin{theorem}\label{T:csstcyclic}
Let $I_1, I_2\subset \zn$ be cyclotomic cosets. Then $(C(I_1),C(I_2))$ is a CSS-T pair if and only if:
\begin{enumerate}
    \item[\rm (1)] $I_2\subset I_1$ and
    \item[\rm (2)] $n\not \in (I_1+I_1+I_2)$.
\end{enumerate}
The parameters of the corresponding quantum code are $[[n,\abs{I_1}-\abs{I_2},\geq n-\Amp(J_2)+1]]$.
\end{theorem}
\begin{proof}
We use the third equivalent condition from Theorem \ref{T:equiv-def} with $C_1=C(I_1)$ and $C_2=C(I_2)$. We have
$$
C(I_2)\subset C(I_1)\iff I_2\subset I_1,
$$
and 
$$
\begin{aligned}
C(I_2)\subset (C(I_1)^{\star 2})^\perp &\iff \1 \in (C(I_1)^{\star 2}\star C(I_2))^\perp=C(I_1+I_1+I_2)^\perp \\
&\iff \1 \in B(-(I_1+I_1+I_2))^\perp \iff n\not \in I_1+I_1+I_2,
\end{aligned}
$$
as follows from (\ref{commutativity}) and \rl{sumschur}. Also, the last equivalence follows from \cite[Prop. 1]{galindostabilizer}. 
We use Corollary \ref{C:dmin2perp} for the parameters of the quantum code.
\end{proof}

\begin{remark}\label{R:condcyclic}
Theorem \ref{T:csstcyclic} also holds if we substitute condition (2) with
\begin{enumerate}
    \item[(2')] $I_1+I_1\subset -J_2$.
\end{enumerate}
This is because
$$
C(I_2)\subset (C(I_1)^{\star 2})^\perp =C(I_1+I_1)^\perp
\iff I_1+I_1\subset -J_2.
$$
As $I_2\subset I_1$, from Theorem \ref{T:csstcyclic}, we obtain the necessary condition $n\not \in I_2$ for $(C(I_1),C(I_2))$ to be a CSS-T pair. This happens if and only if $n\in -J_2$. Hence, if the pair $I_1, I_2$ satisfies the conditions from Theorem \ref{T:csstcyclic}, then the pair $I_1\cup \{n\}, I_2$ also satisfies those conditions. This is a translation of the following fact that we have seen in the previous section: If $(C_1, C_2)$ is a CSS-T pair, then $(C_1+\la\1\ra, C_2)$ is also a CSS-T pair.
\end{remark}

\begin{example}\label{EX:cyc3}
We consider $I_1,I_2$ as in Example \ref{EX:cyc2}. Clearly $I_2\subset I_1$. From the computation of $I_1+I_2$ in Example \ref{EX:cyc2}, we obtain
$$
I_1+I_1+I_2=[n-1]=\{1,2,\dots,14\}.
$$
By Theorem \ref{T:csstcyclic}, we have that $(C(I_1),C(I_2))$ is a CSS-T pair with parameters $[[15,1,3]]$. Note that we have recovered the (punctured) quantum Reed-Muller code mentioned in the introduction.
\end{example}

In Section~\ref{S:poset}, we studied conditions for a CSS-T pair to be maximal in each component. The following result shows how we can translate those conditions to cyclic codes. 

\begin{corollary}\label{C:cycmaximal}
Let $I_1,I_2\subset \zn$ be cyclotomic cosets such that $(C(I_1),C(I_2))$ is a CSS-T pair. Then the pair $(C(I_1),C(I_2))$ is maximal in $C_1$ if and only if 
$$
-J_1=I_2\cup (I_1+I_2),
$$
is maximal in $C_2$ if and only if
$$
-J_2=(-J_1)\cup(I_1+I_1),
$$
and is maximal if and only if $$-J_1=I_1+I_2 \text{~and~} -J_2=I_1+I_1.$$
\end{corollary}
\begin{proof}
The conditions for maximality in $C_1$ and $C_2$ follow from Corollary \ref{C:max-in-C1} and Proposition \ref{P:max-in-C2}, respectively, taking into account Theorem \ref{T:dualcyclic} and Lemma \ref{L:sumschur}. The condition for maximality follows similarly from Theorem \ref{T:maximal}. 
\end{proof}

\begin{example}\label{EX:cyc4}
Continuing with the setting from Example \ref{EX:cyc3}, it is easy to check, using Example \ref{EX:cyc2}, that $-J_1=I_1+I_2$ and $-J_2=I_1+I_1$. Therefore, by Corollary \ref{C:cycmaximal}, the CSS-T pair $(C(I_1),C(I_2))$ is maximal.
\end{example}

From Corollary~\ref{C:dmin2perp}, we see that it is desirable to find CSS-T pairs $(C_1, C_2)$ such that $C_1^{\star 2}$ has a large minimum distance. In \cite{cascudosquarescyclic}, it is shown that the construction of cyclic codes based on the notion of restricted weight can give rise to codes $C$ such that both $C$ and $C^{\star 2}$ have excellent parameters. It is, therefore, interesting to study when we can use these codes for constructing CSS-T pairs. We briefly explain the construction from \cite{cascudosquarescyclic} and then obtain CSS-T codes from this construction. In what follows, we assume that $n=q^s-1$.

\begin{definition}
Let $a\in [n]$ have $q$-ary representation $(a_{s-1},a_{s-2},\dots,a_{0})_q$, and let $1\leq t\leq s$. The $t$-restricted weight of $a$ is defined as
$$
w_q^{(t)}(a):=\max_{i\in \{0,\dots,s-1\}}\sum_{j=0}^{t-1}a_{i+j},
$$
where we consider the sum $i+j$ modulo $s$. In other words, it is the maximum number of nonzero elements for any sequence of $t$ (cyclically) consecutive digits of the $q$-ary representation of $a$.
\end{definition}

The $t$-restricted weight is invariant under multiplication by $q$, and we can speak about the $t$-restricted weight of a minimal cyclotomic coset. It is shown in \cite[Prop. 11]{cascudosquarescyclic} that
$$
w_q^{(t)}(a)\leq w_q^{(t)}(b)+ w_q^{(t)}(c),
$$
for $b,c\in [n]$ and $a=b+c\bmod n$. Therefore, given cyclotomic cosets $I_1,I_2\subset \zn$ whose elements have $t$-restricted weight at most $\mu_1,\mu_2$, respectively, the cyclotomic coset $I_1+I_2$ will have $t$-restricted weight at most $\mu_1+\mu_2$. Let $I^t_{\leq \mu}:=\{a\in \zn : w_q^{(t)}(a)\leq \mu \}$. In \cite[Prop. 13]{cascudosquarescyclic}, it is proven that for $a\in I^t_{\leq \mu}$, we have $w_q^{(s)}(a)\leq \lfloor (\mu s)/t\rfloor$. This motivates the following construction.

\begin{corollary}\label{C:restricted}
Take $1\leq t\leq s$ and $1\leq \mu_1,\mu_2\leq t$. If $\mu_2\leq \mu_1$ and $2\lfloor (\mu_1 s)/t\rfloor+\lfloor (\mu_2 s)/t\rfloor\leq s-1$, then $(C(I^t_{\leq \mu_1}),C(I^t_{\leq \mu_2}))$ is a CSS-T pair.
\end{corollary}
\begin{proof}
We use Theorem \ref{T:csstcyclic} with $I_i=I^t_{\leq \mu_i}$, for $i=1,2$. As $\mu_2\leq \mu_1$, we have $I_2\subset I_1$. We claim that $n\not \in I_1+I_1+I_2$. Indeed, let $z=a+b+c \bmod n$, with $a,b\in I_1$, $c\in I_2$. By the previous discussion, 
$$
w_2^{(s)}(z)=w_2^{(s)}(a+b+c)\leq w_2^{(s)}(a)+w_2^{(s)}(b)+w_2^{(s)}(c)\leq 2\lfloor (\mu_1 s)/t\rfloor+\lfloor (\mu_2 s)/t\rfloor\leq s-1. 
$$
Since $w_2^{(s)}(n)=s$, we conclude that $n\not \in I_1+I_1+I_2$, and the result follows from Theorem \ref{T:csstcyclic}.
\end{proof}

Note that, by Remark~\ref{R:condcyclic}, we can also consider $C_1=C(I^t_{\leq \mu_1}\cup \{n\})$ for the previous result. For the parameters of the corresponding CSS-T code, in \cite{cascudosquarescyclic}, there are formulas for the parameters of $C(I^t_{\leq \mu})$ in some cases, and we can also use the usual bounds for cyclic codes.

\begin{example}\label{EX:cyc5}
It is easy to check that $I_1$ and $I_2$ from Example \ref{EX:cyc2} are precisely
$$
I_1=I^4_{\leq \mu_1}\cup \{15\} \qquad \text{and} \qquad I_2=I^4_{\leq \mu_2}
$$
with $\mu_1=\mu_2=1$. Note that, for $t=s=4$, the conditions from Corollary \ref{C:restricted} are satisfied. Therefore, $(C(I^4_{\leq \mu_1}), C(I^4_{\leq \mu_2}))$ is a CSS-T pair, which implies that $(C(I_1), C(I_2))$ is a CSS-T pair (which we already knew by Example \ref{EX:cyc3}).
\end{example}

\subsection{Extended cyclic codes}
We define $\hzn:=\{0\}\cup \zn$. We will adapt the definitions from the previous section for this setting. Let $I\subset \hzn$. We say that $I$ is a cyclotomic coset if $I=q\cdot I$. For $I_1, I_2\subset \hzn$, we define $I_1+I_2$ as in (\ref{sumcyclo}), where we understand that $i_1+i_2=0$ if and only if $i_1=i_2=0$, for $i_1\in I_1$ and $i_2\in I_2$, and the rest of the sums are computed as usual in $\zn=\{1,\dots,n\}$. We denote by $J:=\hzn\setminus I$. 

For $\mathcal{M}\subset \{0,\dots,n\}$, we consider $\hat{X}_n:=\{0\}\cup X_n$, the zero locus of $x^{n+1}-x$, and we define 
$$
\hat{B}(\mathcal{M}):=\ev_{\hat{X}_n}(\mathcal{L}(\mathcal{M}))=
\{(f(0),f(1),f(\beta),\dots,f(\beta^{n-1})) : f\in  \mathcal{L}(\mathcal{M})\}\subset \mathbb{F}_{q^s}^{n+1}.
$$
For $I\subset\hzn$ a cyclotomic coset, the extended cyclic code associated with $I$ is
$$
\hat{C}(I):=\hat{B}(I)\cap \mathbb{F}_{q}^{n+1}.
$$
Note that in this case, we are not considering $-I$. With respect to the parameters, $\hat{C}(I)$ has parameters $[n+1,\abs{I},\geq n-\max(I)+1]$, and $\hat{C}(I)^\perp$ has parameters $[n+1,n+1-\abs{I},\geq \delta(I)+1]$, where $\delta(I)$ is the maximum number of consecutive elements in $I$ as before (it is a BCH-type bound for extended cyclic codes).

Although these codes are no longer cyclic, they still preserve some of the properties of cyclic codes. The proof of the following result is analogous to the one in \cite[Thm. 1]{cascudosquarescyclic}.

\begin{lemma}
Let $I_1, I_2\subset \hzn$ be cyclotomic cosets. Then
$$
\begin{aligned}
&\hat{C}(I_1)\star \hat{C}(I_2) =\hat{C}(I_1+I_2).
\end{aligned}
$$
\end{lemma}

As a consequence, one can check that Theorem \ref{T:csstcyclic} and Corollary \ref{C:restricted} also hold when we consider extended cyclic codes. Moreover, for extended cyclic codes, one may also allow $\mu_1=0$ or $\mu_2=0$ in Corollary \ref{C:restricted}. When considering the $s$-restricted weight, in \cite[Prop. 10]{cascudosquarescyclic}, it is shown that Corollary \ref{C:restricted} for extended cyclic codes corresponds to the family of CSS-T pairs obtained by using binary Reed-Muller codes from \cite{felicecsst}. Nevertheless, by considering the $t$-restricted weight, with $t<s$, we obtain different families of CSS-T codes. Moreover, considering the general case from Theorem \ref{T:csstcyclic}, it is clear that we obtain a much larger family of CSS-T pairs than by using binary Reed-Muller codes, thus obtaining a wider range of parameters. In the following example, we show that we can improve the parameters of the CSS-T codes obtained with binary Reed-Muller codes in some cases. All the computations from the following examples were done using SageMath \cite{sagemath}. 

\begin{example}\label{E:greedy}
We use a greedy construction to obtain CSS-T codes with cyclic codes, and we compare them with the CSS-T codes obtained with binary Reed-Muller codes. Let $s>1$, $n=2^s-1$, and we consider the cyclotomic cosets associated with the extension $\F_{2^s}/\F_2$. Assume that $\zn=\II_{a_1}\cup \II_{a_2}\cup \cdots \cup \II_{a_\ell}$, with $1=a_1<a_2<\cdots a_\ell$. We consider the following greedy construction: let $I_2 := \II_{a_1}\cup \II_{a_2}\cup \cdots \cup \II_{a_t}$, for some $t< \ell$ such that $n\not \in I_2+I_2+I_2$, and let $I_1^{(0)} := I_2$. If $I'_1 := I_1^{(0)}\cup \II_{a_{t+1}}$ satisfies $n\not \in I'_1+I'_1+I_2$, we set $I_1^{(1)}:= I'_1$, and we set $I_1^{(1)} := I_1^{(0)}$ otherwise. Following this procedure until we cannot add any more minimal cyclotomic cosets, we will get a cyclotomic coset $I_1^{(u)}$, for some $t\leq u <\ell$, such that $n\not \in I_1^{(u)}+I_1^{(u)}+I_2$. Therefore, by Theorem \ref{T:csstcyclic} and Remark \ref{R:condcyclic}, we get that $(C(I_1^{(u)}\cup \{n\}),C(I_2))$ is a CSS-T pair. Moreover, we have the BCH bound
$$
\wt(C(I_2)^\perp)\geq n-\Amp(J_2)+1=\delta(I_2)+1=a_{t+1},
$$
which bounds the minimum distance of the corresponding quantum code by Corollary \ref{C:dmin2perp}. Note that this construction can be easily generalized to extended cyclic codes.

For $s\leq 6$, the CSS-T codes obtained with the previous construction do not improve the parameters of the CSS-T codes obtained with binary Reed-Muller codes. Nevertheless, for $s=7,8,9,10$, we show in Table \ref{tablecsst} that we can obtain a broader range of parameters using cyclic and extended cyclic codes, and some of these codes outperform the ones derived from binary Reed-Muller codes. For all the codes in Tables \ref{tablecsst} and \ref{tablecsst2} we have checked that the bound for the minimum distance is sharp. 

\begin{table}[ht]

\caption{Parameters of the CSS-T codes obtained with binary Reed-Muller, cyclic, and extended cyclic codes (using the greedy construction).} 
\centering
\begin{tabular}{||c|c||c|c||c|c||}
 \hline 
$s$ & Reed-Muller  \\
  \hline \hline
7 & $[[128,21,4]]$  \\

8 & $[[256,84,4]]$ \\

9 & $[[512,120,4]]$ \\
9 & $[[512,84,8]]$ \\

10 & $[[1024,375,4]]$ \\
10 & $[[1024,120,8]]$ \\
\hline
\end{tabular}
\begin{tabular}{||c|c||c|c||c|c||}
 \hline 
$s$ & Cyclic \\
  \hline \hline
7 & $[[127,29,3]]$ \\
7 & $[[127,15,5]]$ \\
7 & $[[127,8,7]]$ \\

8 & $[[255,85,3]]$ \\
8 & $[[255,39,5]]$ \\
8 & $[[255,21,7]]$ \\

9 & $[[511,148,3]]$ \\
9 & $[[511,112,5]]$ \\
9 & $[[511,103,7]]$ \\

10 & $[[1023,376,3]]$ \\
10 & $[[1023,213,5]]$ \\
10 & $[[1023,191,7]]$ \\
10 & $[[1023,161,9]]$ \\
10 & $[[1023,131,11]]$ \\
10 & $[[1023,116,13]]$ \\
10 & $[[1023,106,15]]$ \\
\hline
\end{tabular}
\begin{tabular}{||c|c||c|c||c|c||}
 \hline 
$s$ & Extended cyclic \\
  \hline \hline
7 & $[[128,28,4]]$ \\
7 & $[[128,14,6]]$ \\
7 & $[[128,7,8]]$ \\

8 & $[[256,84,4]]$ \\
8 & $[[256,36,6]]$ \\
8 & $[[256,20,8]]$ \\

9 & $[[512,147,4]]$ \\
9 & $[[512,111,6]]$ \\
9 & $[[512,102,8]]$ \\

10 & $[[1024,375,4]]$ \\
10 & $[[1024,210,6]]$ \\
10 & $[[1024,190,8]]$ \\
10 & $[[1024,160,10]]$ \\
10 & $[[1024,130,12]]$ \\
10 & $[[1024,115,14]]$ \\
10 & $[[1024,105,16]]$ \\
\hline
\end{tabular}
\label{tablecsst}
\end{table}

Using Remark 3.13 from \cite{albertocsst}, it is easy to see that, for $n$ even, if we consider $e_i$, $1\leq i\leq n$, the standard basis vectors in $\F_{2}^n$, and the code
$$
C=\langle e_{2i-1}+e_{2i},\; 1\leq i\leq n/2 \rangle,
$$
then $(C,\la\1\ra)$ is a CSS-T pair with parameters 
\begin{equation}\label{eq:csstdmin2}
[[n,n/2-1,2]].
\end{equation}
This code has better parameters than the CSS-T codes with minimum distance 2 derived from binary Reed-Muller, cyclic, or extended cyclic codes in the cases we have checked. Therefore, we have omitted the codes with minimum distance $2$ from Table \ref{tablecsst} and the ones with dimension 0. 

For a direct comparison, we can see that the CSS-T codes obtained from binary Reed-Muller codes with parameters $[[128,21,4]]$, $[[512,120,4]]$, $[[512,84,8]]$ and $[[1024,120,8]]$ are outperformed by the CSS-T codes derived from extended cyclic codes with parameters $[[128,28, 4]]$, $[[512,147, 4]]$, $[[512,102, 8]]$ and $[[1024,190, 8]]$, respectively.
\end{example}

\begin{example}
Not all the codes from the previous example are maximal with respect to $C_1$. Therefore, it is possible to use our Corollary \ref{C:propagation} to increase the dimension of the corresponding quantum code in some cases. For example, one can check that the CSS-T code with parameters $[[255,21, 7]]$ from Table \ref{tablecsst} is not maximal with respect to the first component using Corollary \ref{C:max-in-C1}. By Proposition \ref{P:step}, this means that there is some vector $y\in C_2^\perp\cap (C_1\star C_2)^\perp$ such that $y\not \in C_1$ and $(C_1+\langle y \rangle, C_2)$ is a CSS-T pair. The parameters of the corresponding quantum code are $[[255,22, 7]]$ by Corollary \ref{C:propagation}, increasing the dimension of the quantum code by 1. By computer search, we have found a vector $y$ such that $(C_1+\langle y \rangle, C_2)$ is still not maximal with respect to the first component. Hence, there is a vector $y'$ such that $(C_1+\langle y,y' \rangle, C_2)$ is a CSS-T pair with parameters $[[255,23, 7]]$, increasing the dimension of the original quantum code by 2. In the cases where we have found such $y,y'$, the pair $(C_1+\langle y,y' \rangle, C_2)$ is maximal with respect to the first component, and we cannot continue to increase the dimension using Corollary \ref{C:propagation}.

In Table \ref{tablecsst2}, we show the codes that can be derived from CSS-T codes using binary Reed-Muller codes, cyclic codes, and extended cyclic codes (with the greedy construction from Example \ref{E:greedy}) by applying Corollary \ref{C:propagation} for length $2^s$, $s=4,\dots,10$ ($2^s-1$ for cyclic codes). All the codes in Table \ref{tablecsst2} are maximal with respect to the first component of the CSS-T pair, although it might be possible to improve them further since there are many choices for the vectors that we add to $C_1$ in Corollary \ref{C:propagation}. We note that the CSS-T codes derived from cyclic and extended cyclic codes still outperform the improved CSS-T codes arising from Reed-Muller codes. The parity check matrices of the classical codes used to construct the quantum codes from Tables \ref{tablecsst} and \ref{tablecsst2} can be found in the GitHub repository {\tt RodrigoSanJose/Cyclic-CSS-T} \cite{githubCSST}. 

\begin{table}[ht]

\caption{Parameters of improved CSS-T codes obtained with binary Reed-Muller, cyclic, and extended cyclic codes (using the greedy construction).} 
\centering
\begin{tabular}{||c|c||c|c||c|c||}
 \hline 
$s$ & Reed-Muller \\
  \hline \hline
5 & $[[32,4,4]]$\\
7 & $[[128,26,4]]$ \\
9 & $[[512,133,4]]$ \\
10 & $[[1024,125,8]]$ \\
\hline
\end{tabular}
\begin{tabular}{||c|c||c|c||c|c||}
 \hline 
$s$ & Cyclic \\
  \hline \hline
5 & $[[31,4,3]]$ \\
8 & $[[255,23,7]]$ \\
9 & $[[511, 149, 3]]$ \\
10 & $[[1023,219,5]]$ \\
10 & $[[1023,193,7]]$ \\
10 & $[[1023,133,11]]$ \\
\hline
\end{tabular}
\begin{tabular}{||c|c||c|c||c|c||}
 \hline 
$s$ & Extended cyclic \\
  \hline \hline
5 & $[[32,4,4]]$ \\
8 & $[[256,22,8]]$ \\
9 & $[[512, 148, 4]]$ \\
10 & $[[1024,217,6]]$ \\
10 & $[[1024,192,8]]$ \\
10 & $[[1024,133,12]]$ \\ 
\hline
\end{tabular}
\label{tablecsst2}
\end{table}
\end{example}

\section{Relation to triorthogonal codes}\label{S:tri}

Another family of codes that is usually studied for fault-tolerant computation, and, in particular, for magic state distillation, are triorthogonal codes \cite{bravyiTriorthogonalOriginal,haahClassificationTriorthogonal}. A binary matrix $G$ of size $m\times n$ is called \textit{triorthogonal} if $\wt(G_a\star G_b)=0 \bmod 2$, for all pairs of rows $1\leq a < b \leq m$, and $\wt(G_a\star G_b\star G_c)=0\bmod 2$, for all triples of rows $1\leq a <b <c \leq m$. With such a matrix, by taking $C_1$ to be the linear span of $G$ and $C_2$ the linear span of the even weighted rows of $G$, one can construct a quantum code (which we will call \textit{triorthogonal code}) such that, when a transversal $T$ gate is applied to it, it induces a transversal $T$ gate on the logical qubits, up to Clifford corrections. This is stronger than having a CSS-T code, since the definition of CSS-T only requires the physical transversal $T$ to induce some logical operation on the logical qubits. If one wants to avoid the Clifford corrections, some weight conditions have to be imposed on the classical codes used (see \cite[Thm. 4]{rengaswamyOptimalityCSST}).
From our results, we can obtain the following. 

\begin{corollary}\label{C:1inC3}
    If $(C_1, C_2)$ is a CSS-T pair, then $\1 \in(C_2^{\star 3})^\perp$.    
\end{corollary}

\begin{proof}
As $C_2\subseteq C_1$, Corollary~\ref{C:monotone} implies that $(C_2,C_2)$ is a CSS-T pair. Thus, $C_2^{\star 2}\subset C_2^\perp$ by Theorem~\ref{T:equiv-def}, meaning that $\1 \in(C_2^{\star 3})^\perp$.
\end{proof} 

Having $\1 \in(C_2^{\star 3})^\perp$ implies that $C_2$ has a triorthogonal generator matrix, which is also the case for triorthogonal codes due to the fact that, in that setting, the generator matrix for $C_2$ is a submatrix of a triorthogonal matrix.

Since the triorthogonality condition is stronger than being CSS-T, it may be possible that CSS-T codes achieve better parameters than triorthogonal codes. To see this, we consider the \textit{scaling exponent} of the distillation protocol presented in \cite{bravyiTriorthogonalOriginal}. They obtain that
$$
\gamma=\frac{\log_2(n/k)}{\log_2(d)},
$$
for an $[[n,k,d]]$ triorthogonal code. Since the distillation overhead scales as $O(\log^\gamma(1/\epsilon))$, where $\epsilon$ is the output accuracy (see \cite{bravyiTriorthogonalOriginal} for details), codes with lower $\gamma$ are preferred. We will use this value for CSS-T codes to compare the goodness of their parameters with some of the triorthogonal codes in the literature. In \cite{bravyiTriorthogonalOriginal}, the authors find a family of triorthogonal codes with parameters $[[3k+8,k,\geq 2]]$, where $k$ is even. The CSS-T codes from (\ref{eq:csstdmin2}) have strictly better parameters. In particular, the scaling exponent $\gamma$ tends to $1$ for the codes in (\ref{eq:csstdmin2}), while the family from \cite{bravyiTriorthogonalOriginal} has scaling exponent tending to $\log_2(3) \approx 1.585$. In \cite{bravyiTriorthogonalOriginal} they also obtain a code with parameters $[[49,1,5]]$, and $\gamma=2.418$. If we compare with the codes in our tables, in particular, the codes $[[32,4,4]]$ and $[[1024,192,8]]$ (to take an example of a short code and a long code), we obtain for $\gamma$ the values $1.5$ and $0.805$, respectively. 

In \cite{haahClassificationTriorthogonal}, the authors find triorthogonal codes with parameters $[[35,3,3]]$ and $[28,2,3]]$, with scaling exponent equal to $2.236$ and $2.402$, respectively, which are higher values than the one we obtained for $[[32,4,4]]$. Moreover, the authors in \cite{haahClassificationTriorthogonal} prove that there is no triorthogonal quantum code with minimum distance larger than $3$ when $n+k\leq 38$, while $[[32,4,4]]$ satisfies these last two conditions (but it is not triorthogonal, only CSS-T). Furthermore, in \cite{haahSublogarithmicOverhead}, triorthogonal codes with $\gamma < 1$ are found, but they require at least $\approx 2^{58}$ qubits. With CSS-T, codes it is possible to find codes with $\gamma<1$ and a much lower number of qubits, for example the code  $[[1024,192,8]]$ we showed before. The shorter CSS-T code that we find with $\gamma<1$ is the code with parameters $[[256,84,4]]$, which has $\gamma=0.804$. This shows that one can indeed obtain better parameters by relaxing the conditions on the classical codes and requiring them to be CSS-T instead of triorthogonal. We reiterate that this discussion is purely in terms of parameters, since triorthogonal codes implement the logical $T$ gate, while for CSS-T codes we only require that they support a transversal $T$ gate.

\section{Conclusion} \label{S:conclusion}

In this paper, we considered binary CSS-T codes, which are quantum stabilizer codes that respect a transversal gate. We provided a straightforward characterization of binary CSS-T codes and used it to demonstrate that CSS-T codes form a poset. We determined maximal and minimal elements of this poset as well as elements which are maximal with respect to one code in a CSS-T pair. We demonstrated a propagation rule for nondegenerate CSS-T codes. We used cyclotomic cosets to characterize CSS-T pairs from cyclic codes. Moreover, we obtained quantum codes with better parameters than those
in the literature, using cyclic and extended cyclic codes. A number of related open problems remain, such as  determining a similar characterizations of $q$-ary CSS-T codes and considering other families of classical codes to construct CSS-T codes. 

\section{Acknowledgements}
Part of this work was done during the visit of Diego Ruano, Rodrigo San-Jos\'e, and Ivan Soprunov to Virginia Tech. They thank Eduardo Camps Moreno, Hiram H. L\'opez, and Gretchen L. Matthews for their hospitality. The initial collaboration amongst the group (absent San-Jos\'e) was facilitated by the Collaborate@ICERM program, supported by the National Science Foundation under Grant No. DMS-1929284.

\section*{Declarations}
\subsection*{Conflict of interest} The authors declare no conflict of interest.

\bibliographystyle{abbrv}

\begin{thebibliography}{10}

\bibitem{andersonquantumRM}
J.~T. Anderson, G.~Duclos-Cianci, and D.~Poulin.
\newblock Fault-tolerant conversion between the {S}teane and {R}eed-{M}uller quantum codes.
\newblock {\em Phys. Rev. Lett.}, 113:080501, Aug 2014.

\bibitem{felicecsst}
E.~Andrade, J.~Bolkema, T.~Dexter, H.~Eggers, V.~Luongo, F.~Manganiello, and L.~Szramowski.
\newblock {CSS-T} codes from {R}eed {M}uller codes for quantum fault tolerance.
\newblock {\em ArXiv 2305.06423}, 2023.

\bibitem{cod_package}
T.~Ball, E.~Camps, H.~Chimal-Dzul, D.~Jaramillo-Velez, H.~L\'{o}pez, N.~Nichols, M.~Perkins, I.~Soprunov, G.~Vera-Mart\'{\i}nez, and G.~Whieldon.
\newblock Coding theory package for {M}acaulay2.
\newblock {\em J. Softw. Algebra Geom.}, 11(1):113--122, 2021.

\bibitem{albertocsst}
E.~Berardini, A.~Caminata, and A.~Ravagnani.
\newblock Structure of {CSS} and {CSS-T} quantum codes.
\newblock {\em Des. Codes Cryptogr.}, 2024.

\bibitem{bierbrauercyclic}
J.~Bierbrauer.
\newblock The theory of cyclic codes and a generalization to additive codes.
\newblock {\em Des. Codes Cryptogr.}, 25(2):189--206, 2002.

\bibitem{bravyiTriorthogonalOriginal}
S.~Bravyi and J.~Haah.
\newblock Magic-state distillation with low overhead.
\newblock {\em Phys. Rev. A}, 86:052329, Nov 2012.

\bibitem{calderbankp}
A.~R. Calderbank, E.~M. Rains, P.~W. Shor, and N.~J.~A. Sloane.
\newblock Quantum error correction via codes over {${\rm GF}(4)$}.
\newblock {\em IEEE Trans. Inform. Theory}, 44(4):1369--1387, 1998.

\bibitem{CalderbankShor_96}
A.~R. Calderbank and P.~W. Shor.
\newblock Good quantum error-correcting codes exist.
\newblock {\em Phys. Rev. A}, 54:1098--1105, Aug 1996.

\bibitem{githubCSST}
E.~Camps-Moreno, H.~H. L\'{o}pez, G.~L. Matthews, D.~Ruano, R.~San-Jos\'{e}, and I.~Soprunov.
\newblock Parity check matrices for the codes in {``An algebraic characterization of binary CSS-T codes and cyclic CSS-T codes for quantum fault tolerance''. GitHub repository}.
\newblock Available online: {\tt https://github.com/RodrigoSanJose/Cyclic-CSS-T}, 2024.
\newblock Accessed on 18 April 2024.

\bibitem{cascudosquarescyclic}
I.~Cascudo.
\newblock On squares of cyclic codes.
\newblock {\em IEEE Trans. Inform. Theory}, 65(2):1034--1047, 2019.

\bibitem{cascudosquaresmpc}
I.~Cascudo, J.~S. Gundersen, and D.~Ruano.
\newblock Squares of matrix-product codes.
\newblock {\em Finite Fields Appl.}, 62:101606, 21, 2020.

\bibitem{galindostabilizer}
C.~Galindo, F.~Hernando, and D.~Ruano.
\newblock Stabilizer quantum codes from {$J$}-affine variety codes and a new {S}teane-like enlargement.
\newblock {\em Quantum Inf. Process.}, 14(9):3211--3231, 2015.

\bibitem{grasslsurvey}
M.~Grassl.
\newblock Algebraic quantum codes: linking quantum mechanics and discrete mathematics.
\newblock {\em International Journal of Computer Mathematics: Computer Systems Theory}, 6(4):243--259, 2021.

\bibitem{grasslp}
M.~Grassl.
\newblock New quantum codes from {CSS} codes.
\newblock {\em Quantum Inf. Process.}, 22(1):Paper No. 86, 11, 2023.

\bibitem{Mac2}
D.~R. Grayson and M.~E. Stillman.
\newblock {M}acaulay2, a software system for research in algebraic geometry.

\bibitem{haahSublogarithmicOverhead}
M.~B. Hastings and J.~Haah.
\newblock Distillation with sublogarithmic overhead.
\newblock {\em Phys. Rev. Lett.}, 120:050504, Jan 2018.

\bibitem{haahClassificationTriorthogonal}
S.~Nezami and J.~Haah.
\newblock Classification of small triorthogonal codes.
\newblock {\em Phys. Rev. A}, 106(1):Paper No. 012437, 13, 2022.

\bibitem{quanQuantumRM}
D.-X. Quan, L.-L. Zhu, C.-X. Pei, and B.~C. Sanders.
\newblock Fault-tolerant conversion between adjacent {R}eed-{M}uller quantum codes based on gauge fixing.
\newblock {\em J. Phys. A}, 51(11):115305, 16, 2018.

\bibitem{rainsp}
E.~M. Rains.
\newblock Nonbinary quantum codes.
\newblock {\em IEEE Trans. Inform. Theory}, 45(6):1827--1832, 1999.

\bibitem{calderbankclassicalcsst}
N.~Rengaswamy, R.~Calderbank, M.~Newman, and H.~D. Pfister.
\newblock Classical coding problem from transversal {T} gates.
\newblock In {\em 2020 IEEE International Symposium on Information Theory (ISIT)}, pages 1891--1896, 2020.

\bibitem{rengaswamyOptimalityCSST}
N.~Rengaswamy, R.~Calderbank, M.~Newman, and H.~D. Pfister.
\newblock On optimality of {CSS} codes for transversal {T}.
\newblock {\em IEEE Journal on Selected Areas in Information Theory}, 1(2):499--514, 2020.

\bibitem{Shor_95}
P.~W. {Shor}.
\newblock {Scheme for reducing decoherence in quantum computer memory}.
\newblock {\em Phys. Rev. A}, 52(4):R2493--R2496, Oct. 1995.

\bibitem{Steane_96}
A.~{Steane}.
\newblock {Multiple-Particle Interference and Quantum Error Correction}.
\newblock {\em Proceedings of the Royal Society of London Series A}, 452(1954):2551--2577, Nov. 1996.

\bibitem{sagemath}
{The Sage Developers}.
\newblock {\em {S}ageMath, the {S}age {M}athematics {S}oftware {S}ystem ({V}ersion 10.3)}, 2023.
\newblock {\tt https://www.sagemath.org}.

\end{thebibliography}

\end{document}